\documentclass{article}
\PassOptionsToPackage{numbers,sort}{natbib}   
\usepackage{amsmath,amssymb,amsthm,mathtools}   

\newif\ifnipssty
\IfFileExists{neurips_2026.sty}{\usepackage[preprint]{neurips_2026}\nipsstytrue}{}
\ifnipssty
  \workshoptitle{ML$\times$OR: Mathematical Foundations and Operational Integration of Machine Learning for Uncertainty-Aware Decision-Making}
\else
  \usepackage[letterpaper,textwidth=5.5in,textheight=9in]{geometry}
  \usepackage{times}
  \renewcommand{\And}{\and}
  \date{}
\fi

\usepackage{booktabs,multirow}
\usepackage{xcolor}
\usepackage[numbers,sort]{natbib}
\setcitestyle{numbers,square}
\usepackage[colorlinks=true,linkcolor=black,citecolor=black,urlcolor=blue]{hyperref}

\theoremstyle{plain}
\newtheorem{theorem}{Theorem}
\newtheorem{proposition}{Proposition}
\newtheorem{corollary}{Corollary}
\theoremstyle{definition}

\AtBeginDocument{%
  \setlength{\abovedisplayskip}{4pt plus 1pt minus 2pt}%
  \setlength{\belowdisplayskip}{4pt plus 1pt minus 2pt}%
  \setlength{\abovedisplayshortskip}{2pt}%
  \setlength{\belowdisplayshortskip}{2pt}}
\newcommand{\tself}[1]{\tau^{\mathrm{self}}_{#1}}
\newcommand{\Nb}{\mathcal{N}}
\newcommand{\Gr}{\mathcal{I}}
\newcommand{\R}{\mathbb{R}}
\newcommand{\E}{E}
\newcommand{\Var}{\mathrm{var}}
\newcommand{\sech}{\operatorname{sech}}

\title{Spillover Effects under Network Interference\\ When Neighbours' Treatment Effects Are Heterogeneous}

\author{%
  Faezeh Dehghan Tarzjani\\
  University of Southern California\\
  \texttt{dehghant@usc.edu}
  \And
  Bhaskar Krishnamachari\\
  University of Southern California\\
  \texttt{bkrishna@usc.edu}
}

\begin{document}
\maketitle

\begin{abstract}
Optimising budget-constrained network interventions requires evaluating not just who is connected, but predicting how strongly individual recipients will propagate the treatment's benefits. Existing models predict spillover from neighbours' treatments and attributes. We argue that spillover also depends on how strongly each neighbour responded to its own treatment. We prove that no model in which neighbours' responsiveness enters separately from their treatments can capture this interaction, and we introduce \textsc{SpilloverNet}, a graph neural network designed to preserve neighbour-level response dynamics. Because a neighbour's response is unobserved, a natural approach is to estimate it from covariates and plug it in. However, we prove that any predictor relying solely on standard network data faces an irreducible error floor set by unobserved personal responsiveness. Empirically, at low heterogeneity the plug-in's correlation with the true response looks acceptable while its spillover error is already 24.5\% against an oracle 13.8\%; as heterogeneity grows its error climbs to 51.9\%, worse than using no responsiveness estimate at all. A per-unit estimate from a direct-response measurement, collected in a small pilot before spillover arrives, escapes this bound and cuts regret against oracle-$\tau$ targeting by up to 14 percentage points in a budgeted targeting problem. On two real social graphs, \textsc{SpilloverNet} reaches 7.7--8.4\% error, outperforming standard GNNs as well as specialised causal-representation baselines.
\end{abstract}

\section{The problem}

Decision-makers in public health, development economics and digital platforms allocate
scarce treatments over a network and rely on diffusion to stretch a fixed budget: a few
vaccinations to protect a community, microloans seeded with well-connected farmers, a
handful of rewarded users to start a cascade. Each allocation rests on a prediction: the
spillover $S_i$, the benefit unit $i$ receives from its treated neighbours. Error in
that prediction becomes policy regret \citep{ma2021}: formally, the planner maximises
$\sum_i\tau_iZ_i+S_i(Z)$ over allocations $Z$ subject to a budget
(eq.~\ref{eq:decision}, Section~\ref{sec:exp}). The interference literature offers estimands and
estimators for $S_i$ \citep[e.g.][]{hudgens2008,aronow2017,forastiere2021,leung2022},
reviewed in \citet{schweinberger2025}; all predict it from \emph{which} neighbours were
treated.

When responsiveness varies this is not enough. Suppose two of a unit's neighbours are
treated. One responds strongly and passes the benefit on; the other barely responds and
passes on almost nothing. A rule that counts treated neighbours scores the two the
same, and will spend budget on well-connected units that transmit little. What matters
is the sender's own response,
\begin{equation}
\tself{j}=Y_j(1,0_{-j})-Y_j(0,0_{-j}),
\label{eq:tauself}
\end{equation}
and we study models where $S_i$ depends on $\{\tself{j}: j\in\Nb(i)\}$. The difficulty
is that $\tself{j}$ is not observed when the allocation is made, and, as we show,
estimating it from covariates does not work. A two-stage design does: a pilot treats a
random subset and records an intermediate outcome $Y^{\mathrm{mid}}_j$ (early
engagement, say, or an early antibody titre) after treatment but before spillover
arrives. This gives a per-unit estimate $\hat\tau_j$ that retains individual
heterogeneity.

The paper makes three points. First, no existing model expresses sender-level response
(Section~\ref{sec:gap}). Second, a covariate-based estimate cannot stand in for the
unobserved response: any such predictor has an irreducible error floor, which we
quantify (Section~\ref{sec:floor}). Third, a pilot estimate can. We build
\textsc{SpilloverNet}, a graph neural network with a per-neighbour gate that accepts a
per-unit response estimate (Section~\ref{sec:est}), and show that targeting on the
pilot estimate cuts regret against oracle-$\tau$ targeting by up to $14$ points relative
to a degree rule (Section~\ref{sec:exp}). The closest prior work is \citet{zhang2025},
who let the spillover coefficient vary with the sender's community, and
\citet{leungloupos2025}, who place GNNs inside the interference framework; neither has
a channel for the sender's individual response.

\section{Why existing models cannot express this}
\label{sec:gap}

\textbf{Setup.} Units $i\in V$ on a network with neighbourhoods $\Nb(i)$; covariates
$X_i$, binary treatment $Z_i$, outcome $Y_i$. Treatment is random given the network
(seeding plus diffusion is allowed), so there is no confounding. Responsiveness is a
fixed trait that shows up only under treatment:
\begin{equation}
\tau_i=\tau_0+\varphi(X_i)+\eta_i,\qquad
\eta_i\sim N(0,\sigma^2_\eta)\ \text{i.i.d.},\quad \eta_i\perp X_i .
\label{eq:tau}
\end{equation}
Here $\tau_i=\tself{i}$ from \eqref{eq:tauself}; $\tilde\tau_j=(\tau_j-\bar\tau)/s_\tau$
is its within-graph standardisation, and in population
$\sigma^2_\tau=\Var\{\varphi(X)\}+\sigma^2_\eta$. Covariates can recover $\varphi(X_i)$;
they cannot recover $\eta_i$. Outcomes follow $Y_i=\mu(X_i)+\tau_iZ_i+S_i+\varepsilon_i$,
with spillover a saturating sum over treated neighbours, each scaled by its own
responsiveness,
{\small
\begin{equation}
S_i\approx\psi(R_i),\qquad
R_i=\textstyle\sum_{j\in\Nb(i)}\theta^b_{ij}\,Z_j\,\{1+\alpha_\tau\tanh\tilde\tau_j\},
\label{eq:spill}
\end{equation}}
plus a small two-hop term (full model in Supplement~A). Here
$\psi(u)=\operatorname{sign}(u)\sqrt{|u|}$, the odd square root, since $R_i$ can be
negative; $\bar R_i=\E(R_i\mid\Gr)$ is its mean given the network information $\Gr$
defined below. The saturation gives diminishing returns in treated neighbours. The
particular $\psi$ and $\tanh$ are chosen for concreteness;
Proposition~\ref{prop:nonsep} needs only a non-constant modulation, and the floor in
Theorem~\ref{thm:floor} exists for any smooth one.

\textbf{The gap.} Existing models predict $S_i$ from neighbours' treatments and
attributes, across the assignment-vector \citep{hudgens2008}, exposure-mapping
\citep{aronow2017,leung2022}, mediation \citep{vanderweele2013,kundu2026} and
structural \citep{manski1993,zhang2025} strands (paper-by-paper record in
Supplement~C). None of them has a channel for $\tself{j}$, and one cannot simply be
added as another attribute. Call a model \emph{separable} if neighbours' responses
enter it only through a summary of their own,
$F\{s(z_{\Nb(i)},X_{\Nb(i)}),\,s_\tau(\tilde\tau_{\Nb(i)})\}$, with $s,s_\tau$
invariant to relabelling neighbours; attribute-dependent weights such as
$\sum_j\beta(X_j)Z_j$ \citep{zhang2025} are allowed inside $s$.

\begin{proposition}[Non-separability]
\label{prop:nonsep}
The model above is not separable, for any choice of $s$, $s_\tau$, $F$.
\end{proposition}

The proof takes three lines (Supplement~B). Take a unit with two neighbours who look
identical on covariates, and treat first the stronger responder, then the weaker. A
separable model sees the same inputs in both cases, one treated neighbour and the same
two people, so it must predict the same spillover; the true spillover differs. Every
model in Supplement~C is either separable, with responsiveness absent, or reads the
sender's response only through its realised outcome $Y_j$, where it is mixed with
$S_j$ and $\varepsilon_j$ (see Limitations).

\section{Any covariate-based predictor has irreducible error}
\label{sec:floor}

The result in this section concerns the information available, not any particular
estimator. Let $\Gr$ be the information in $(X,Z,G)$. If covariates leave individual
variation in responsiveness unexplained ($\sigma_\eta>0$), every predictor restricted
to $\Gr$ has irreducible error. The natural fix, estimating responsiveness from
covariates with a conditional-average-treatment-effect (CATE) method and plugging the
estimate in, is one such predictor: every CATE method targets $\E(\tau\mid X)$, and
that expectation removes $\eta$.

\begin{theorem}[Floor]
\label{thm:floor}
Any predictor $\hat S_i$ that uses only $(X,Z,G)$ satisfies
$\E\{(\hat S_i-S_i)^2\}\ge V^\star_i:=\E\{\Var(S_i\mid\Gr)\}$, and for the model above
\begin{equation}
V^\star_i\;\approx\;(\sigma^2_\eta/\sigma^2_\tau)\,\alpha_\tau^2\,(4|\bar R_i|)^{-1}
\textstyle\sum_{j\in\Nb(i)}Z_j\,(\theta^b_{ij})^2\sech^4(m_j),
\qquad m_j=\E(\tilde\tau_j\mid X_j),
\label{eq:floor}
\end{equation}
where $\bar R_i=\E(R_i\mid\Gr)$ as in \eqref{eq:spill}, for $\bar R_i$ bounded away
from $0$, up to the threshold factor $c_i$ and a smaller two-hop term (exact form and
proof in Supplement~B).
\end{theorem}

\begin{corollary}
\label{cor:cate}
Substituting any estimate of $\E(\tau\mid X)$ for $\tau$, applied to units it was not
fit on, gives spillover error at least $V^\star_i$, no matter the sample size or the
estimator.
\end{corollary}

The floor grows with the unexplained share $\sigma_\eta/\sigma_\tau$. It is roughly
flat in treated degree, because the saturation $\psi$ turns the sum into a weighted
mean of $\theta^b_{ij}$. The units that matter most for targeting are therefore those
whose neighbours' responses the covariates predict worst. Since the generator is known,
we compute $V^\star_i$ exactly by resampling $\eta$ with $(X,Z,G)$ fixed
(Table~\ref{tab:decomp}).

\section{SpilloverNet: a graph neural network with a normalised gate}
\label{sec:est}

After $L$ message-passing layers, node $i$'s vector $h^{(L)}_i$ summarises its $L$-hop
ball. Our spillover depends only on the two-hop ball (Supplement~B, Prop.~B.2), so two
layers reach every argument of $S_i$. We train four, which fit better at the degrees we
study; since the target is a fixed two-hop function, no fixed-point condition is needed
at any depth. Two structural terms of $\theta^b_{ij}$ cannot be computed by message
passing from our inputs (Supplement~B, Prop.~B.8; Limitations). Standard GNNs treat all
neighbours alike: GraphSAGE averages them and a GCN sums them with degree weights.
Neither can capture a mechanism whose point is that neighbours differ. SpilloverNet
separates a unit's own state from its neighbours' messages and gives each neighbour its
own weight:
{\small
\begin{equation}
h^{(\ell)}_i=\sigma\big(U_\ell\,h^{(\ell-1)}_i+W_\ell\textstyle\sum_{j\in\Nb(i)}\alpha^{(\ell)}_{ij}\,h^{(\ell-1)}_j\big),
\qquad
\alpha^{(\ell)}_{ij}=\rho_\ell\,\operatorname{softmax}_{j\in\Nb(i)}\big(e^{(\ell)}_{ij}\big).
\label{eq:layer}
\end{equation}}
$U_\ell$ carries $i$'s own state and $W_\ell$ the aggregated messages, so at every
layer the neighbours' contribution, which is the spillover, has its own parameters. The
decoupling follows the D-GCN of \citet{dehghan2025dgcn}, built for wireless
interference, where throughput depends on each neighbour's own access probability
\citep{tarzjani2025pcsma}. The score $e^{(\ell)}_{ij}$ is learned from $j$'s vector,
$i$'s vector and the edge features, and measures how much neighbour $j$ matters to $i$.
The softmax turns the scores into weights, and the learned scalar $\rho_\ell\in(0,1)$
makes them sum to less than one at every node; we call the result a \emph{normalised
gate}. Supplement~F compares both choices with graph attention
\citep{velickovic2018,brody2022}: decoupling the self path is worth $13$--$18$ NMAE
points, and $\rho_\ell<1$ a further $3.1$ points at Flickr's mean degree of $63$,
though it costs $0.4$--$0.9$ points at degree $5$--$10$.

Node inputs are $(X_i,\,d_i,\,Z_i,\,\tau_i)$, with $\tau_i$ replaced by an estimate in
the feasible regime. Training minimises squared error on $S_i$, and we report
$\mathrm{NMAE}=\sum_i|\hat S_i-S_i|/\sum_i|S_i|$. Setting $\alpha_{ij}=1/d_i$ removes
the gate; we call this the Mean comparator. Supplement~E examines a contraction
certificate, which the trained models do not meet.

\section{Experiments}
\label{sec:exp}

\textbf{Protocol.} $500$ graphs of $100$ nodes from three random-graph families;
$\varphi(X)=0.3X_1+0.2\,1\{X_2>0\}$, so $\Var\{\varphi(X)\}=0.10$; $10\%$ seeding plus
three diffusion rounds; $200$ epochs; five seeds; $\sigma_\eta\in\{0.1,0.3,0.6,1.0,1.5\}$,
over which $\sigma_\eta/\sigma_\tau$ runs $0.30\to0.98$. Paired $t$-tests by seed;
decision rules fixed before each run.

\textbf{SpilloverNet uses the response channel.} With true $\tau$, SpilloverNet reaches
$13.8$--$14.1\%$ NMAE, flat in $\sigma_\eta$. Shuffling $\tau$ across nodes raises the
error to $59$--$63\%$ ($p<10^{-6}$), and removing the gate costs $0.7$--$1.9$ points,
more at larger $\sigma_\eta$ ($p\le0.03$; Supplement~H).

\begin{table}[t]
\caption{Feasible regime, five seeds. Floor: NMAE of the conditional median, the
$L_1$ analogue of Theorem~\ref{thm:floor}, by Monte Carlo ($50$ graphs, $300$ resamples
of $\eta$). No-$\tau$: SpilloverNet with the response input zeroed. $r$: correlation
between estimated and true $\tau$.}
\label{tab:decomp}
\centering\scriptsize\setlength{\tabcolsep}{4pt}
\begin{tabular}{lccccc}
\toprule
$\sigma_\eta$ & $0.1$ & $0.3$ & $0.6$ & $1.0$ & $1.5$\\
\midrule
Floor & $14.1$ & $30.3$ & $37.5$ & $40.2$ & $41.1$\\
No-$\tau$ SpilloverNet & $21.2$ & $34.9$ & $42.1$ & $44.4$ & $45.6$\\
CATE plug-in & $24.5$ & $40.0$ & $47.5$ & $50.5$ & $51.9$\\
\quad $r(\hat\tau,\tau)$ & $0.67$ & $0.49$ & $0.30$ & $0.17$ & $0.10$\\
Direct-response plug-in & $20.9$ & $29.9$ & $27.5$ & $22.8$ & $19.1$\\
\quad $r(\hat\tau,\tau)$ & $0.24$ & $0.28$ & $0.34$ & $0.37$ & $0.39$\\
Oracle $\tau$ & $13.8$ & $14.1$ & $14.0$ & $13.8$ & $13.8$\\
\bottomrule
\end{tabular}
\end{table}

\textbf{The smoothness trap.} Replacing true $\tau$ with a CATE estimate more than
doubles the error across the sweep (Table~\ref{tab:decomp}), and it does so by two
routes. At $\sigma_\eta\le0.3$, where $\sigma^2_\eta\le\Var\{\varphi(X)\}=0.10$, the
usual diagnostic misses the problem. The correlation $r(\hat\tau,\tau)$ is computed
over the whole population, where $\tau$ varies mostly through $\varphi(X)$, which CATE
recovers ($r=0.67$ and $0.49$). Spillover, however, depends on differences
\emph{within} a neighbourhood, and those are mostly $\eta$, which CATE removes; the
error is $24.5$ and $40.0\%$ against $13.8$ and $14.1\%$ with true $\tau$. At
$\sigma_\eta\ge0.6$, $\eta$ dominates and the diagnostic does catch it ($r\le0.30$),
but the plug-in is still what the model is fed. The floor and No-$\tau$ rows show what
this costs. At $\sigma_\eta=1.5$ no covariate-only method can beat $41.1\%$, and
SpilloverNet with the response input zeroed reaches $45.6\%$, close to the bound. The
CATE plug-in reaches $51.9\%$, worse than using no responsiveness estimate at all, and
it is worse at every $\sigma_\eta$: the gate learns to trust a channel that carries no
within-neighbourhood signal.

\textbf{A per-unit estimate escapes the floor.} The bound covers only predictors built
from $(X,Z,G)$. A \emph{direct-response} measurement is an outcome recorded after
treatment but before spillover arrives, $Y^{\mathrm{mid}}_j=\mu(X_j)+\tau_jZ_j+\varepsilon'_j$;
for treated units $\hat\tau_j=Y^{\mathrm{mid}}_j-Y^{\mathrm{pre}}_j$, with
$Y^{\mathrm{pre}}$ a pre-treatment baseline. The estimate is per person, so it keeps
$\eta_j$, and it contains no $S_j$. It beats the CATE plug-in by $4$ to $33$ points at
every $\sigma_\eta$, and its correlation with $\tau$ rises with $\sigma_\eta$
($0.24\to0.39$), the opposite of CATE's, as we predicted before the runs
(Supplement~I). For $\sigma_\eta\ge0.6$ it is below the floor, at $27.5\%$ against
$37.5\%$ at $0.6$ and $19.1\%$ against $41.1\%$ at $1.5$, which is possible only
because $Y^{\mathrm{mid}}$ lies outside $(X,Z,G)$. Corollary~\ref{cor:cate} can be
read constructively: it says what has to be measured. A pre-period difference, which
carries $S_j$, is $1.4$--$2.9$ points worse for $\sigma_\eta\ge0.6$ and within about
one seed s.d.\ of the direct-response estimate below that (Supplement~D).

\begin{table}[t]
\caption{Targeting under a $10\%$ budget: regret (\%) of \eqref{eq:decision} relative
to the oracle-$\tau$ policy, mean $\pm$ s.e.\ over $50$ graphs. Last column: mean true
$\tau$ (and degree) of the units selected at $\sigma_\eta=1.5$; population mean
$\tau\approx0.6$. Budgets $20$--$30\%$: Supplement~G.}
\label{tab:target}
\centering\scriptsize\setlength{\tabcolsep}{3.5pt}
\begin{tabular}{lcccccc}
\toprule
$\sigma_\eta$ & $0.1$ & $0.3$ & $0.6$ & $1.0$ & $1.5$ & $\bar\tau_{\rm sel}$ ($\bar d_{\rm sel}$)\\
\midrule
Oracle $\tau$ & $0$ & $0$ & $0$ & $0$ & $0$ & $2.96$ ($8.8$)\\
Direct-response $\hat\tau$ & $29.9\pm1.9$ & $30.0\pm1.6$ & $29.8\pm1.7$ & $30.8\pm1.8$ & $32.9\pm1.8$ & $1.49$ ($9.6$)\\
CATE $\hat\tau$ & $29.4\pm2.3$ & $36.1\pm2.5$ & $40.6\pm2.3$ & $43.9\pm2.3$ & $46.5\pm2.1$ & $1.08$ ($8.8$)\\
Degree only & $33.0\pm2.2$ & $34.5\pm2.1$ & $37.7\pm2.2$ & $41.9\pm2.4$ & $46.6\pm2.6$ & $0.61$ ($12.5$)\\
Random & $57.2\pm2.0$ & $59.8\pm2.1$ & $62.7\pm2.1$ & $65.8\pm2.1$ & $69.0\pm2.1$ & $0.55$ ($6.5$)\\
\bottomrule
\end{tabular}
\end{table}

\textbf{Targeting under a budget.} The decision problem is
\begin{equation}
\textstyle\max_{Z\in\{0,1\}^N}\ \sum_i\tau_iZ_i+S_i(Z)\quad\text{s.t.}\quad\sum_iZ_i\le B,
\label{eq:decision}
\end{equation}
and Table~\ref{tab:target} reports regret. Every policy scores node $j$ by
$\hat\tau_j(1+d_j)$ and treats the top $B$. The policies differ only in where
$\hat\tau$ comes from, and regret is measured against the same rule run with true
$\tau$ rather than against the optimum of \eqref{eq:decision}; every chosen $Z$ is
evaluated in the true generator. The pilot seeds a random $10\%$ and lets treatment
diffuse. Units it treated get $\hat\tau_j$ from their own outcome and all others get
the pilot-treated mean, so outside the pilot every policy ranks by degree. Using one
heuristic for all policies means the regret gap isolates the value of the
responsiveness signal; putting SpilloverNet inside a welfare maximiser is left for
future work. Direct-response targeting holds regret near $30\%$ while degree-only rises
from $33\%$ to $47\%$, so measuring response is worth $3$ to $14$ points, more at
higher heterogeneity (the $3$ at $\sigma_\eta=0.1$ is within noise). CATE targeting is
\emph{no better than ignoring responsiveness} for $\sigma_\eta\ge0.3$, and the last
column shows why. The degree policy picks hubs with average response ($0.6$) at mean
degree $12.5$. CATE finds slightly better responders ($1.08$) but drops to degree
$8.8$. Direct response finds responders more than twice as strong ($1.49$) at degree
$9.6$, about the oracle's reach. The gap narrows as the budget grows, from $13.7$
points at $10\%$ to $8.2$ at $30\%$ (Supplement~G). The $30\%$ regret that every
feasible policy shares is the cost of coverage: the pilot reveals $\tau$ only where it
treated.

\begin{table}[t]
\caption{Real graphs, generated outcomes; NMAE (\%), mean (s.d.) over $25$ runs
(five realisations $\times$ five seeds). Attention variants on Flickr: Supplement~F.}
\label{tab:real}
\centering\scriptsize
\begin{tabular}{lcc}
\toprule
Model & Flickr ($n\!=\!7{,}575$, $\bar d\!\approx\!63$) & BlogCatalog ($n\!=\!5{,}196$, $\bar d\!\approx\!66$)\\
\midrule
SpilloverNet & $7.73$ ($0.65$) & $8.36$ ($0.63$)\\
GraphSAGE & $10.84$ ($1.22$) & $10.94$ ($0.69$)\\
MLP (no graph) & $18.65$ ($1.82$) & $17.69$ ($1.24$)\\
NetDeconf \citep{guo2020} & $19.68$ ($1.79$) & $20.92$ ($2.11$)\\
NetEst \citep{jiang2022} & $20.00$ ($2.32$) & $20.89$ ($1.95$)\\
\bottomrule
\end{tabular}
\end{table}

\textbf{Real graphs.} Table~\ref{tab:real} repeats the comparison on two real networks
with generated outcomes. The ranking is the same on both, up to the tie between
NetDeconf and NetEst. Both lose to a graph-free MLP; they are built to remove
confounding, and randomised treatment has none. When every unit responds the same, as
in the generator of \citet{zhang2025}, the gate has nothing to learn and SpilloverNet
matches Mean, as we predicted (Supplement~C).

\textbf{Limitations.} All outcomes, including those on the real topologies, are
generated by our model. We know of no dataset with observed responsiveness on a
network, and showing the mechanism in real outcomes is the main open item. A variant
that feeds neighbours' observed outcomes directly fails a scrambling test, which is the
reflection problem \citep{manski1993,angrist2014} in network form, and we exclude it
from all claims. Two structural terms of $\theta^b_{ij}$, the common-neighbour count
and the clustering coefficient, cannot be computed by message passing from the
synthetic inputs (Supplement~B, Prop.~B.8); we have not tested whether supplying them
as features closes the $14\%$ oracle-$\tau$ residual on a noise-free target.

\bibliographystyle{plainnat}

\clearpage
\appendix
\section*{Supplementary Material}

\section{Generative model: full specification}
\label{app:dgp}
Base weights
$\theta^b_{ij}=c_{\mathrm b}+c_{\mathrm s}e^{-\|X_i-X_j\|}+c_{\mathrm c}(C_i+C_j)/2
+c_{\mathrm n}|\Nb(i)\cap\Nb(j)|/(d_i+1)$ with $C_i$ the local clustering coefficient and
$(c_{\mathrm b},c_{\mathrm s},c_{\mathrm c},c_{\mathrm n})=(0.02,0.50,0.20,0.35)$;
$\mu(X_i)=0.5X_{i1}+0.3X_{i2}^2+0.2\tanh(X_{i1}X_{i2})$; threshold factor $c_i=1.3$ when
the treated fraction in $\Nb(i)$ exceeds one half, else $1$; two-hop term
$S^{(2)}_i=(|D_i|+1)^{-1}\sum_{k\in D_i}\lambda Z_k\{1+\beta_\tau\tanh\tilde\tau_k\}$ with
$D_i=B_2(i)\setminus B_1(i)$. The full spillover is
$S_i=\psi(c_iR_i)+\lambda_2S^{(2)}_i$ with $R_i$ as in \eqref{eq:spill}; constants
$\tau_0=0.5$, $\alpha_\tau=1.5$, $\beta_\tau=0.8$, $\lambda=0.6$, $\lambda_2=0.1$,
$\sigma_\varepsilon=0.3$. Diffusion
adopts an untreated node with probability $0.30$ times its treated neighbour share.

\section{Proofs}\label{app:proofs}

\subsection{Proof of Proposition \ref{prop:nonsep} (full)}
Suppose separable $s,s_\tau,F$ exist. Fix a node $i$ with exactly two neighbours $j,k$,
$X_j=X_k$, $\theta^b_{ij}=\theta^b_{ik}=\theta>0$, $D_i=\emptyset$, so $S_i=\psi(R_i)$.
In both configurations below exactly one neighbour is treated, so the treated fraction
is one half and $c_i=1$. Let $a\neq b$ and consider (I) $Z_j=1,Z_k=0$ and (II)
$Z_j=0,Z_k=1$, with $\tilde\tau_j=a,\tilde\tau_k=b$ in both. The treatment--covariate
multiset is $\{(1,X_j),(0,X_j)\}$ and the response multiset is $\{a,b\}$ in both; by
permutation invariance $s$ and $s_\tau$ agree across (I) and (II), so $F$ does, and
$S^{\mathrm{I}}_i=S^{\mathrm{II}}_i$. But
$S^{\mathrm{I}}_i=\psi\{\theta(1+\alpha_\tau h(a))\}$ and
$S^{\mathrm{II}}_i=\psi\{\theta(1+\alpha_\tau h(b))\}$. Since $h$ (here $\tanh$) is non-constant we may
choose $a,b$ with $h(a)\neq h(b)$; $\alpha_\tau\neq0$ and $\theta>0$ give distinct
arguments; $\psi$ injective gives distinct values. Contradiction. \hfill$\square$

\subsection{Proposition B.2 (strict approximate neighbourhood interference) and proof}
\textit{Statement.} In the model of Section~2 and Supplement~A, if $z,z'$ agree on the
two-hop ball $B_2(i)$ then $Y_i(z)=Y_i(z')$ exactly; hence every argument of $S_i$ lies
in $B_2(i)$, and two message-passing layers reach all of them (what they can compute
from them is Proposition~B.8).

\textit{Proof.} 
Fix $i$ and let $z,z'$ agree on $B_2(i)$. $S_i$ depends on $z$ only through
$\{Z_j:j\in\Nb(i)\}$, in $R_i$ and in $c_i$, and $\{Z_k:k\in D_i\}$, in $S^{(2)}_i$;
both index sets lie in $B_2(i)$, so $S_i(z)=S_i(z')$. The remaining dependence of $Y_i$
on $z$ is through $\tau_iZ_i$, and $i\in B_2(i)$. The terms $\mu(X_i),\tau_i,
\varepsilon_i$ are pre-treatment. Hence $Y_i(z)=Y_i(z')$ exactly. \hfill$\square$

\subsection{Theorem B.3 (layerwise contraction) and proof}
\textit{Statement.} Write $T_\ell(h)=\sigma(W_\ell\mathcal A_\ell h)$ with
$(\mathcal A_\ell h)_i=\sum_j\alpha^{(\ell)}_{ij}h_j$, $\alpha^{(\ell)}_{ij}\ge0$,
$\sum_j\alpha^{(\ell)}_{ij}=\rho_\ell$, $\sigma$ $1$-Lipschitz. Then
$\|T_\ell(h)-T_\ell(h')\|_\infty\le\rho_\ell\|W_\ell\|\,\|h-h'\|_\infty$; if
$\kappa=\max_\ell\rho_\ell\|W_\ell\|<1$ the weight-tied iteration has a unique fixed
point with geometric convergence at rate $\kappa$.

\textit{Proof.} 
For any $i$, by the triangle inequality and $\alpha_{ij}\ge0$,
$\|\sum_j\alpha_{ij}(h_j-h'_j)\|\le\sum_j\alpha_{ij}\|h_j-h'_j\|\le\rho_\ell\|h-h'\|_\infty$.
Applying $W_\ell$ multiplies by at most $\|W_\ell\|$ and $\sigma$ is $1$-Lipschitz, giving
the displayed bound. If $\kappa<1$, $T$ is a $\kappa$-contraction on the complete space
$(\R^{N\times d},\|\cdot\|_\infty)$ and the Banach fixed-point theorem gives existence,
uniqueness and the geometric rate. With a self-term $U_\ell h$ the modulus becomes
$\rho_\ell\|W_\ell\|+\|U_\ell\|$. \hfill$\square$

\subsection{Theorem B.4 (co-estimation contraction) and proof}
\textit{Statement.} With $\psi=\mathrm{id}$, $c_i\equiv1$ and no two-hop term, let
$\Phi:\tau\mapsto\tau'$ form $S_i(\tau)$ and set $\tau'_j=Y_j-\mu(X_j)-S_j(\tau)$ at
treated $j$, with $\tau'_j=\tau_j$ at untreated $j$, whose $\tau_j$ enters no $S_i$.
Then $\Phi$ is Lipschitz in $\|\cdot\|_\infty$ with modulus
$\kappa_\tau=(\alpha_\tau/\sigma_\tau)\theta_{\max}\max_jd^{(1)}_j$, $d^{(1)}_j$ the
number of treated neighbours of $j$; if $\kappa_\tau<1$ the iteration converges
geometrically to a unique fixed point.

\textit{Proof.} 
Under the simplifications, $S_i(\tau)=\sum_{j\in\Nb(i)}\theta^b_{ij}Z_j
\{1+\alpha_\tau\tanh(\tilde\tau_j)\}$. Using $|\tanh'|\le1$ and
$\theta^b_{ij}\le\theta_{\max}$,
$|S_i(\tau)-S_i(\tau')|\le(\alpha_\tau/\sigma_\tau)\theta_{\max}\sum_{j\in\Nb(i)}Z_j
|\tau_j-\tau'_j|\le(\alpha_\tau/\sigma_\tau)\theta_{\max}d^{(1)}_i\|\tau-\tau'\|_\infty$.
At treated $j$, $|\Phi(\tau)_j-\Phi(\tau')_j|=|S_j(\tau)-S_j(\tau')|$, and at untreated
$j$ the difference is zero; taking the maximum over $j$ gives modulus $\kappa_\tau$, and
Banach applies when $\kappa_\tau<1$. With saturation the chain rule contributes
$\psi'(R_i)=1/(2|R_i|^{1/2})$, unbounded as $R_i\to0$. When $i$ has no treated
neighbour $R_i\equiv0$ and nothing depends on $\tau$; the problematic case is
cancellation, $\tanh\tilde\tau_j<-1/\alpha_\tau$ for some treated $j$. A modulus
$\kappa_\tau/(2\varsigma_0^{1/2})$ exists on $\{\min_{i:\,d^{(1)}_i>0}|R_i|\ge\varsigma_0\}$.
\hfill$\square$

\subsection{Proposition B.5 (error propagation) and proof}
\textit{Statement.} Under the conditions of B.4, if
$N^{-1}\sum_j(\hat\tau_j-\tau_j)^2\le\delta^2$ then
$N^{-1}\sum_i(\hat S_i-S_i)^2\le(\alpha_\tau/\sigma_\tau)^2\theta^2_{\max}\,d_{\max}^2\,\delta^2$,
$d_{\max}$ the maximum degree.

\textit{Proof.} 
Write $e_j=\hat\tau_j-\tau_j$ and $K=(\alpha_\tau/\sigma_\tau)\theta_{\max}$. By B.4
and Cauchy--Schwarz, $(\hat S_i-S_i)^2\le K^2d_i\sum_{j\in\Nb(i)}e_j^2$. Averaging over
$i$ and exchanging summation,
\[
N^{-1}\sum_id_i\sum_{j\in\Nb(i)}e_j^2=N^{-1}\sum_je_j^2D_j,\qquad
D_j=\sum_{i\in\Nb(j)}d_i\le d_jd_{\max}\le d_{\max}^2,
\]
whence $N^{-1}\sum_i(\hat S_i-S_i)^2\le K^2d_{\max}^2\,\delta^2$. If degree is
uncorrelated with $e_j^2$ the factor $d_{\max}^2$ may be replaced by
$\langle d^2\rangle=N^{-1}\sum_id_i^2$, the second moment of degree; this is a
heuristic, not part of the bound. \hfill$\square$

\textit{Remark.} Theorems~B.4--B.5 characterise the iterative co-estimation alternative
that SpilloverNet does \emph{not} take; the main text (Section~\ref{sec:est}) uses
supervised regression on a fixed two-hop ball, for which no fixed-point condition is
required (Prop.~B.2). They are included to show what the iterative route would cost,
not to describe the model we use.

\subsection{Proof of Theorem \ref{thm:floor}}
The first claim is the $L^2$ projection identity
$\E\{(\hat S_i-S_i)^2\}=\E[\{\hat S_i-\E(S_i\mid\Gr)\}^2]+\E\{\Var(S_i\mid\Gr)\}$.
For \eqref{eq:floor}, conditional on $\Gr$ the only randomness in $S_i$ is $\eta$, since
$\mu,\theta^b,c_i,\Nb(i),D_i$ are $\Gr$-measurable. Write $s=\sigma_\tau$ and
$\tilde\tau_j=m_j+u_j$ with $m_j$ $\Gr$-measurable, $u_j=\eta_j/s$,
$\E(u_j\mid\Gr)=0$, $\Var(u_j\mid\Gr)=\sigma^2_\eta/s^2$, independent across $j$; the
graph-level statistics $\bar\tau,s$ are treated as constants, their $O(N^{-1/2})$
fluctuation being absorbed in the remainder. Expanding $g_j=\tanh(\tilde\tau_j)$ about
$m_j$, $\Var(g_j\mid\Gr)=\sech^4(m_j)\sigma^2_\eta/s^2+O(\sigma^3_\eta)$. Since the $u_j$
are independent across $j$ given $\Gr$, cross terms in
$R_i=\sum_j\theta^b_{ij}Z_j(1+\alpha_\tau g_j)$ vanish and
$\Var(R_i\mid\Gr)=\alpha^2_\tau(\sigma^2_\eta/s^2)\sum_jZ_j(\theta^b_{ij})^2
\sech^4(m_j)+O(\sigma^3_\eta)$. The delta method with $\psi'(u)=1/(2|u|^{1/2})$ about
$c_i\bar R_i$ gives $\Var\{\psi(c_iR_i)\mid\Gr\}=c_i\Var(R_i\mid\Gr)/(4|\bar R_i|)+O(\sigma^3_\eta)$,
which is \eqref{eq:floor} times the threshold factor $c_i\in\{1,1.3\}$.
Since $D_i\cap\Nb(i)=\emptyset$, $\{\eta_k\}_{k\in D_i}$ is independent of
$\{\eta_j\}_{j\in\Nb(i)}$ given $\Gr$ and the covariance between $\psi(c_iR_i)$ and
$S^{(2)}_i$ is zero; the same expansion inside $S^{(2)}_i$ (Supplement~A) gives
\[
\Var\{\lambda_2S^{(2)}_i\mid\Gr\}
=\frac{\lambda_2^2\lambda^2\beta_\tau^2}{(|D_i|+1)^2}\cdot\frac{\sigma^2_\eta}{s^2}
\sum_{k\in D_i}Z_k\sech^4(m_k)+O(\sigma^3_\eta),
\]
which is the two-hop term omitted from \eqref{eq:floor}. Summing and substituting
$s=\sigma_\tau$ yields \eqref{eq:floor}. Because the generator is known,
$V^\star_i$ may be computed exactly by resampling $\eta$ with $(X,Z,G)$ held fixed.
\hfill$\square$

\subsection{Proof of Corollary \ref{cor:cate}}
An estimator of $\E(\tau_j\mid X_j)$ fit on units other than $j$ enters the predictor as
$\hat f(X_j)$ with $\hat f$ fixed given its training data, so the plug-in predictor is
$\Gr$-measurable and Theorem~\ref{thm:floor} applies. $V^\star_i$ depends on the joint
law of $(S,X,Z,G)$ only, not on $N$ or the estimator; and $\E(\eta_j\mid X_j)=0$ by (2),
so $\E(\tau_j\mid X_j)$ carries none of $\eta_j$. An estimator fit on $j$ itself reads
$Y_j$, which carries $\eta_j$; the corollary excludes that case. \hfill$\square$

\subsection{Proposition B.8 (expressivity) and proof sketch}
\textit{Statement.} $S_i$ is a permutation-invariant function of $(z,w,G)$ restricted
to $B_2(i)$, including the edges inside $B_2(i)$. If $|\Nb(i)\cap\Nb(j)|/(d_i+1)$ is
supplied as an edge feature and $C_i$ as a node feature, a two-layer message-passing
network with injective aggregation and universal update maps approximates $S_i$
uniformly on compact domains. From the inputs of Supplement~H alone it does not: the
two structural terms of $\theta^b_{ij}$ are triangle counts, which message passing
cannot compute \citep{xu2019,chen2020}.

\textit{Proof sketch.} 
By Proposition~B.2, $S_i$ depends on $(z,w,G)$ only through $B_2(i)$, invariant to
relabelling within each hop. With the two structural terms given as features, every
argument of $R_i$, $c_i$ and $S^{(2)}_i$ is a continuous function of node and edge
features at distance at most two, so a two-layer network with injective aggregation
separates whatever $S_i$ separates \citep{xu2019}, and uniform approximation follows
from universality of the update maps and the Stone--Weierstrass theorem on the algebra
of permutation-invariant functions. Without them, message passing is bounded by the
one-dimensional Weisfeiler--Lehman test \citep{xu2019}, which cannot count triangles
\citep{chen2020}; with continuous $X$ the count $|\Nb(i)\cap\Nb(j)|$ is a discontinuous
function of the feature multisets, outside the uniform closure of continuous
permutation-invariant functions. \hfill$\square$

\section{Paper-by-paper record}\label{app:record}

Throughout, $Z_j$ is neighbour $j$'s assignment, $g_i(Z_{\Nb_i})$ an exposure summary,
$Y_j$ a realised or counterfactual outcome, and $\tself{j}$ as in \eqref{eq:tauself}.

\smallskip\noindent\textit{Assignment-vector models.}
\citet{sobel2006}: $Y_i(z)$, and $Y_{ij}(z_i)$ under partial interference; neighbours
enter through assignments received. \citet{rosenbaum2007}: $y_i(z)$ with the uniformity
trial $\{y_i(0)\}$ as baseline; randomisation inference over $z$. \citet{hudgens2008}:
$Y_{ij}(z_i)$ and allocation-strategy averages $\bar Y_{ij}(z;\psi)$; the indirect effect
changes others' assignments while holding $i$ untreated. \citet{tchetgen2012}:
$Y_{ij}(a_{ij},a_{i,-j})$ with individual direct effect
$Y_{ij}(1,a_{i,-j})-Y_{ij}(0,a_{i,-j})$. \citet{manski2013}: $y^j(t_J)$ on the full
vector, and a structural form with peer mean treatment and peer mean outcome.
\citet{puelz2022}: randomisation tests for exposure contrasts, with the conditioning
set chosen from a biclique decomposition of the assignment--exposure graph; the exposure
is a function of $z$. In each, this is not a model of the form $Y_i=F(\tself{j})$.

\smallskip\noindent\textit{Exposure mappings.}
\citet{aronow2017}: $D_i=f(Z,\theta_i)$ with the consistency restriction that equal
exposures give equal outcomes; the exposure is induced by the assignment vector.
\citet{papadogeorgou2019}: cluster potential outcomes averaged over an allocation
programme $P_{\alpha,L}$. \citet{forastiere2021}: $Y_i(Z_i,G_i)$ with
$G_i=g_i(Z_{\Nb_i})$ and estimands $\tau(g),\delta(g;z)$. \citet{leung2020,leung2022}:
exposure summaries, and decaying dependence on assignments at distance. \citet{hu2022}: average direct effect
$\tau_{\mathrm{ADE}}$ and average indirect effect $\tau_{\mathrm{AIE}}=n^{-1}\sum_i
\sum_{j\neq i}\E\{Y_j(w_i=1;W_{-i})-Y_j(w_i=0;W_{-i})\}$; the indirect effect is the
cross-unit effect of changing $i$'s assignment on $j$. \citet{vazquezbare2023}:
$Y_{ig}(d,s)$ with $s=\sum_{j\neq i}D_{jg}$ the number of treated peers under
exchangeability; identification through $\E[Y_{ig}\mid D_{ig}=d,H_{ig}=h]=\E[Y_{ig}(d,h)]$.
In each, this is
not a model of the form $Y_i=F(\tself{j})$.

\smallskip\noindent\textit{Mediation and contagion.}
\citet{vanderweele2012}: contagion $\E\{Y_{i2}(0,Y_{i1}(1))-Y_{i2}(0,Y_{i1}(0))\}$ and
infectiousness $\E\{Y_{i2}(1,Y_{i1}(1))-Y_{i2}(0,Y_{i1}(1))\}$; the neighbour enters through
a counterfactual outcome, not the contrast. \citet{vanderweele2013}: $Y_{ijk}(t,m,g)$ with
$g$ a peer-mediator summary; in its linear specification $g$ reduces to additive
functions of peer treatments and covariates. \citet{ogburn2017}: contagion and
infectiousness through the alter's counterfactual infection status. \citet{imai2020}:
$Y_{i2}\{Z_i,Y^*_{i1}(Z_i)\}$ through a latent post-treatment intention.
\citet{kundu2026}: the REN-SEM with seven pathways; its linear case (LREN-SEM) has an
outcome equation in which neighbour exposure ($\beta_2$) and neighbour covariates
($\beta_6$) enter additively through the degree-normalised operator $S_{1i}$. In each,
this is not a model of the form $Y_i=F(\tself{j})$; the LREN-SEM and the linear
specification of \citet{vanderweele2013} are separable in the sense of Section~2, while
the general REN-SEM, with nonlinear $f_y$ and neighbour mediators that carry each
neighbour's own response, belongs with the outcome-reading models noted below.

\smallskip\noindent\textit{Structural and latent-type.}
\citet{manski1993}: linear-in-means through $\E(y\mid x)$ and peer characteristics.
\citet{hoshino2023}: $Y_j^{(d_j,d_{-j})}$ with equilibrium treatments and marginal
treatment responses. \citet{bhattacharya2020}: chain-graph models in which $Y_i$ depends on neighbours'
treatments and outcomes through undirected edges, with the network itself possibly
misspecified; the neighbour enters through $Z_j$ and $Y_j$, not the contrast.
\citet{krauth}: $y_i(\tau_i,\{\tau_j\})$ with latent types and
pairwise primitive $PE(\tau_i,\tau_j)$; the closest paper, but $\tau_i$ is a type, not
$Y_i(1)-Y_i(0)$. \citet{zhang2025}: $Y_i=\beta_0+\sum_{j\neq i}\beta_{ji}G_{ji}Z_j+\gamma
Z_i+\varepsilon_i$ with $\beta_{ji}=\beta_{C_j,C_i}$ and constant $\gamma$; heterogeneity
is community-level and the direct effect does not vary. On this generator SpilloverNet
and Mean are indistinguishable ($1.33\%$ vs $1.85\%$ NMAE, $p=0.41$), as predicted
before the run: a gate gains nothing when everyone responds the same. In each, this is
not a model of the form $Y_i=F(\tself{j})$.

\smallskip\noindent\textit{Separability.}
Of the models above, those that read $Y_j$---\citet{manski1993,manski2013,bhattacharya2020}
and the contagion models of \citet{vanderweele2012,ogburn2017,imai2020}---are not
separable in the sense of Section~2: $Y_j=\mu(X_j)+\tau_jZ_j+S_j+\varepsilon_j$ carries
the product $\tau_jZ_j$, but only together with $S_j$ and $\varepsilon_j$, which is the
reflection problem of Section~5. The rest are separable, with responsiveness absent; in
\citet{zhang2025} the weight varies with the sender's community, an attribute, which
separability permits.


\section{Pre-period versus direct-response estimators}
\label{app:preperiod}
Both estimators are per-unit and both carry $\eta_j$. The pre-period difference
$Y_j-Y^{\mathrm{pre}}_j-\bar\Delta_{\mathrm{untreated}}$ contains $S_j-\bar S$; the
direct-response difference $Y^{\mathrm{mid}}_j-Y^{\mathrm{pre}}_j$ does not. NMAE (\%)
of SpilloverNet with each substituted for $\tau$, five seeds; $r$ is the correlation
of $\hat\tau$ with $\tau$.

\begin{center}\small
\begin{tabular}{lccccc}
\toprule
$\sigma_\eta$ & $0.1$ & $0.3$ & $0.6$ & $1.0$ & $1.5$\\
\midrule
Pre-period NMAE & $20.2$ & $29.7$ & $30.4$ & $25.2$ & $20.5$\\
\quad $r$ & $0.19$ & $0.22$ & $0.28$ & $0.33$ & $0.36$\\
Direct-response NMAE & $20.9$ & $29.9$ & $27.5$ & $22.8$ & $19.1$\\
\quad $r$ & $0.24$ & $0.28$ & $0.34$ & $0.37$ & $0.39$\\
\bottomrule
\end{tabular}
\end{center}
The predictions fixed before the runs were that direct-response error would be no higher at
every $\sigma_\eta$ and $r$ higher at every $\sigma_\eta$. The second held throughout;
the first held for $\sigma_\eta\ge0.6$, while at $0.1$ and $0.3$ pre-period was lower by
$0.7$ and $0.2$ points, within about one seed s.d.\ ($0.6$ and $0.8$), where $S_j$ is
small relative to $\tau_j$. Both estimators share the
hump shape, which reflects the crossover between the covariate channel and the measured
channel rather than either estimator's noise.

\section{Stability: normalisation and contraction}
\label{app:two}
Write $T_\ell(h)=\sigma(W_\ell\mathcal A_\ell h)$ with $(\mathcal A_\ell h)_i=\sum_j
\alpha^{(\ell)}_{ij}h_j$. If $\sigma$ is $1$-Lipschitz and $\sum_j\alpha^{(\ell)}_{ij}
=\rho_\ell$, then $\|T_\ell(h)-T_\ell(h')\|_\infty\le\rho_\ell\|W_\ell\|\|h-h'\|_\infty$,
and $\kappa=\max_\ell\rho_\ell\|W_\ell\|<1$ gives a unique fixed point for the weight-tied
iteration (proof as in B.3; with the self path the modulus is
$\rho_\ell\|W_\ell\|+\|U_\ell\|$, so the $\kappa$ reported below, which omits
$\|U_\ell\|$, is a lower bound on the modulus and its failure is robust). The
factorisation separates gate normalisation ($\rho<1$,
free) from full contraction ($\kappa<1$). On $25$ trained checkpoints from the full
grid (five $\sigma_\eta$ levels, five seeds), $\rho_\ell\approx0.51$ at every layer
while $\|W_\ell\|$ grows $1.96\to3.15$ with depth. Computing $\kappa$ per checkpoint
as $\max_\ell\rho_\ell\|W_\ell\|$ and then averaging gives $\kappa\in[1.61,1.90]$,
mean $1.74$, $0/25$ certified. A reduced-scale run ($150$ graphs of $60$ nodes, $60$
epochs, three seeds, $15$ checkpoints) shows the same signature at smaller magnitude:
$\rho_\ell\approx0.50$, $\|W_\ell\|$ from $1.75$ to $2.9$, $\kappa\in[1.39,1.56]$,
mean $1.50$, $0/15$ certified, with no trend in $\sigma_\eta$. The certificate fails
through the weight norms, not the gate, and the failure grows with scale. Because the
model is a fixed two-hop regression (Prop.~B.2) nothing rests on it; we report it to
separate the two conditions, normalisation and contraction, that classical models
impose together.

\section{Attention with and without a decoupled self path}
\label{app:gat}
Synthetic grid ($500\times100$, $200$ epochs, five seeds), oracle $\tau$. GATv2 uses
four heads and standard self-loops (own state inside the softmax). GATv2$+$self adds a
separate linear self path per layer and removes the self-loop, mirroring the
$U_\ell$ term of \eqref{eq:layer}. Mean NMAE (\%) over five seeds; the Flickr column
gives mean (s.d.) over $25$ runs. GATv2$+$self and GraphSAGE coincide on Flickr to two
decimals ($10.840$ vs $10.838$); they are distinct runs. GATv2 on BlogCatalog completed
$11$ of $25$ runs and is not reported.

\begin{center}\small
\begin{tabular}{lcccccc}
\toprule
 & \multicolumn{5}{c}{Synthetic, $\sigma_\eta$} & Flickr\\
\cmidrule(lr){2-6}
 & $0.1$ & $0.3$ & $0.6$ & $1.0$ & $1.5$ & ($\bar d\approx63$)\\
\midrule
GATv2 (no self path) & $27.0$ & $27.0$ & $27.1$ & $26.7$ & $27.1$ & $28.3$ ($2.49$)\\
GATv2 $+$ decoupled self path & $13.2$ & $13.7$ & $13.3$ & $12.9$ & $12.9$ & $10.8$ ($1.59$)\\
GraphSAGE (has a root weight) & $15.4$ & $15.5$ & $15.6$ & $15.7$ & $15.9$ & $10.8$ ($1.22$)\\
SpilloverNet & $13.8$ & $14.1$ & $14.0$ & $13.8$ & $13.8$ & $7.7$ ($0.65$)\\
\bottomrule
\end{tabular}
\end{center}
Two predictions were fixed before the runs. (i) That GATv2$+$self would be within two
points of GraphSAGE on the synthetic grid. It did better than that, landing
$0.4$--$0.9$ points below SpilloverNet at every $\sigma_\eta$. The flat $27\%$ of the
undecoupled model across $\sigma_\eta$ indicates it does not exploit the responsiveness
channel at all when own state competes with neighbours in the softmax. (ii) That if the
row sum $\rho_\ell<1$ matters at high degree, GATv2$+$self would land above SpilloverNet
on Flickr. It does, by $3.1$ points ($n=25$ each), with the same inputs and self path.
So the normalised gate costs $0.4$--$0.9$ points at degree $5$--$10$ and gains $3.1$
points at degree $63$.

\section{Targeting at larger budgets}
\label{app:budgets}
Same protocol as Table~\ref{tab:target}; regret (\%) $\pm$ s.e.\ over $50$ graphs, and in
parentheses the mean true $\tau$ of the selected units.

\begin{center}\scriptsize\setlength{\tabcolsep}{3pt}
\begin{tabular}{llccccc}
\toprule
Budget & Policy & $\sigma_\eta=0.1$ & $0.3$ & $0.6$ & $1.0$ & $1.5$\\
\midrule
$20\%$ & Direct-response & $25.0\pm1.4$ ($0.66$) & $27.0\pm1.3$ ($0.68$) & $28.6\pm1.4$ ($0.77$) & $32.2\pm1.5$ ($0.90$) & $35.1\pm1.5$ ($1.10$)\\
 & Pre-period & $25.6\pm1.5$ ($0.65$) & $26.8\pm1.4$ ($0.67$) & $29.2\pm1.5$ ($0.75$) & $31.1\pm1.5$ ($0.89$) & $35.4\pm1.5$ ($1.06$)\\
 & CATE & $25.9\pm1.8$ ($0.73$) & $30.9\pm1.6$ ($0.72$) & $36.5\pm1.5$ ($0.76$) & $41.5\pm1.7$ ($0.82$) & $45.8\pm1.5$ ($0.92$)\\
 & Degree only & $27.2\pm1.7$ ($0.61$) & $30.0\pm1.6$ ($0.60$) & $34.8\pm1.7$ ($0.60$) & $40.0\pm1.8$ ($0.60$) & $45.6\pm1.9$ ($0.59$)\\
 & Eigenvector & $31.6\pm1.6$ ($0.59$) & $33.8\pm1.5$ ($0.59$) & $37.7\pm1.6$ ($0.59$) & $42.4\pm1.7$ ($0.59$) & $47.6\pm1.8$ ($0.58$)\\
 & Random & $49.0\pm1.5$ ($0.60$) & $51.1\pm1.4$ ($0.59$) & $54.5\pm1.5$ ($0.58$) & $58.3\pm1.6$ ($0.57$) & $62.4\pm1.7$ ($0.56$)\\
\midrule
$30\%$ & Direct-response & $22.1\pm1.2$ ($0.65$) & $23.9\pm1.1$ ($0.67$) & $26.6\pm1.1$ ($0.73$) & $31.2\pm1.3$ ($0.81$) & $34.6\pm1.2$ ($0.97$)\\
 & Pre-period & $22.6\pm1.2$ ($0.64$) & $24.1\pm1.2$ ($0.66$) & $27.0\pm1.2$ ($0.71$) & $30.2\pm1.2$ ($0.81$) & $34.7\pm1.3$ ($0.94$)\\
 & CATE & $23.6\pm1.3$ ($0.71$) & $28.2\pm1.2$ ($0.71$) & $34.4\pm1.2$ ($0.74$) & $38.6\pm1.2$ ($0.80$) & $44.2\pm1.1$ ($0.86$)\\
 & Degree only & $23.4\pm1.2$ ($0.61$) & $26.1\pm1.2$ ($0.61$) & $31.1\pm1.3$ ($0.60$) & $36.8\pm1.4$ ($0.60$) & $42.8\pm1.5$ ($0.60$)\\
 & Eigenvector & $26.8\pm1.3$ ($0.60$) & $29.4\pm1.3$ ($0.60$) & $34.1\pm1.3$ ($0.59$) & $39.5\pm1.4$ ($0.59$) & $45.3\pm1.5$ ($0.58$)\\
 & Random & $42.0\pm1.3$ ($0.61$) & $44.9\pm1.2$ ($0.60$) & $49.4\pm1.2$ ($0.59$) & $54.0\pm1.3$ ($0.56$) & $58.9\pm1.3$ ($0.54$)\\
\bottomrule
\end{tabular}
\end{center}
At the $10\%$ budget the pre-period estimator gives regret $30.1$, $31.1$, $31.6$,
$32.3$, $35.0$ across $\sigma_\eta$ (selected $\bar\tau=1.35$, $\bar d=10.2$), between
direct-response and CATE. The oracle is the $\hat\tau_j(1+d_j)$ rule with true $\tau$, not the optimum of
\eqref{eq:decision}; every policy's chosen $Z$ is scored by plugging it into the true
generator, so oracle regret is zero by definition at every budget; its selected $\bar\tau$ is $0.96$,
$1.10$, $1.43$, $1.92$, $2.57$ at $20\%$ and $0.91$, $1.03$, $1.31$, $1.73$, $2.28$ at
$30\%$. The gap between direct-response and degree-only targeting narrows as the budget
grows ($13.7\to10.5\to8.2$ points at $\sigma_\eta=1.5$ for $10$, $20$, $30\%$), because
a larger budget leaves less room for selectivity to matter. CATE $\ge$ degree-only holds
for $\sigma_\eta\ge0.3$ at every budget (tie at $1.5$) and reverses at $0.1$ for the
$10$ and $20\%$ budgets (tie at $30\%$).
Direct-response $\le$ pre-period in $12$ of $15$ cells; the three exceptions are within
one standard error.

\section{Experimental details}
\label{app:details}

\textbf{Synthetic graphs.} $500$ graphs of $100$ nodes, cycling three families:
Erd\H{o}s--R\'enyi with $p\sim U(0.05,0.10)$; Barab\'asi--Albert with
$m\in\{2,3,4\}$; Watts--Strogatz with $k\in\{4,\dots,7\}$ and rewiring probability
$U(0.2,0.4)$. Covariates $X_{i1}\sim N(0,1)$, $X_{i2}\sim U(-1,1)$; node degree is
appended as a third input. Treatment: $10\%$ random seeding, then three diffusion
rounds in which an untreated node adopts with probability $0.30$ times its treated
neighbour share. Graphs are split into training, validation and test sets with a fixed
seed; model selection uses validation NMAE and all reported numbers are on test graphs.

\textbf{Training.} AdamW, learning rate $10^{-3}$, weight decay $10^{-5}$,
ReduceLROnPlateau (factor $0.5$, patience $15$), $200$ epochs, squared-error loss on
$S_i$; model selection on validation NMAE. Four hidden layers of width $128$ for
SpilloverNet, GraphSAGE and the GATv2 variants; GATv2 uses four heads. The gate score
$e^{(\ell)}_{ij}$ is a two-layer MLP on
$[h_i,\,h_j,\,\mathrm{edge}_{ij}]$, where $\mathrm{edge}_{ij}$ collects pairwise
quantities of the two endpoints: the absolute differences of their first covariate and
of their degree, and both treatments $Z_i,Z_j$.
MLP: same widths, no graph. NetDeconf: four-layer GCN encoder (width $128$), MLP head
on $[h_i,Z_i,\hat\tau_i]$, Wasserstein-1 IPM between treated and control
representations with weight $10^{-3}$. NetEst: same encoder with two adversarial
discriminators and gradient reversal. Regulariser strengths were swept and the
default reported; stronger values degraded both baselines monotonically.

\textbf{Estimators of $\tau$.} CATE plug-in: an R-learner \citep{nie2021} fit on
$(X_i,Z_i,Y_i)$ of the training graphs and applied to held-out graphs; Theorem~1 makes
the floor method-independent, and a second CATE method is left for future work. Direct response:
$\hat\tau_j=Y^{\mathrm{mid}}_j-Y^{\mathrm{pre}}_j$ with
$Y^{\mathrm{pre}}_j=\mu(X_j)+\varepsilon^{\mathrm{pre}}_j$ and
$Y^{\mathrm{mid}}_j=\mu(X_j)+\tau_jZ_j+\varepsilon'_j$,
$\varepsilon^{\mathrm{pre}},\varepsilon'$ drawn independently of $\varepsilon$ from the
same noise law as the outcome equation. Pre-period:
$Y_j-Y^{\mathrm{pre}}_j$ minus the mean change among untreated units.

\textbf{Real graphs.} Flickr and BlogCatalog in the preprocessed form of
\citet{guo2020}: node inputs are the leading principal components of the original
node attributes, plus standardised log-degree and clustering coefficient. Outcomes are
generated by the model of Supplement~A on the real topology with $10\%$ seeding and
the same diffusion as above; a ``realisation'' redraws $(\tau,Z,\varepsilon)$ on the
fixed graph (realisation seeds $100,\dots,104$). Five realisations $\times$ five
training seeds per model; the split seed is $42$.

\textbf{Floors and targeting.} Floors: $50$ held-out graphs, $300$ resamples of $\eta$
with $(X,Z,G)$ fixed. Targeting: $50$ graphs; the pilot seeds a random $10\%$ and
diffuses as above; the scoring rule and oracle are as in Section~\ref{sec:exp}.

\textbf{Compute.} All runs used a single Google Colab GPU.

\textbf{Ablations per $\sigma_\eta$.} NMAE (\%) of SpilloverNet with true $\tau$, with
$\tau$ shuffled across nodes within each graph, and with the gate removed
($\alpha_{ij}=1/d_i$); mean (s.d.) over five seeds.

\begin{center}\small
\begin{tabular}{lccccc}
\toprule
$\sigma_\eta$ & $0.1$ & $0.3$ & $0.6$ & $1.0$ & $1.5$\\
\midrule
SpilloverNet, true $\tau$ & $13.8$ ($0.5$) & $14.1$ ($0.7$) & $14.0$ ($0.3$) & $13.8$ ($0.4$) & $13.8$ ($0.4$)\\
Shuffled $\tau$ & $62.6$ ($3.1$) & $60.2$ ($1.0$) & $59.5$ ($1.0$) & $59.6$ ($0.8$) & $59.8$ ($1.3$)\\
No gate (Mean) & $14.7$ ($0.4$) & $14.8$ ($0.4$) & $15.2$ ($0.2$) & $15.3$ ($0.2$) & $15.7$ ($0.4$)\\
\bottomrule
\end{tabular}
\end{center}

\textbf{Seed variability for Table~\ref{tab:decomp}.} Mean (s.d.) over five seeds.

\begin{center}\small
\begin{tabular}{lccccc}
\toprule
$\sigma_\eta$ & $0.1$ & $0.3$ & $0.6$ & $1.0$ & $1.5$\\
\midrule
No-$\tau$ & $21.2$ ($0.3$) & $34.9$ ($0.8$) & $42.1$ ($0.7$) & $44.4$ ($0.7$) & $45.6$ ($0.9$)\\
CATE plug-in & $24.5$ ($0.5$) & $40.0$ ($0.9$) & $47.5$ ($1.4$) & $50.5$ ($1.5$) & $51.9$ ($1.8$)\\
Direct-response plug-in & $20.9$ ($0.6$) & $29.9$ ($0.8$) & $27.5$ ($0.4$) & $22.8$ ($0.5$) & $19.1$ ($0.4$)\\
Oracle $\tau$ & $13.8$ ($0.5$) & $14.1$ ($0.7$) & $14.0$ ($0.3$) & $13.8$ ($0.4$) & $13.8$ ($0.4$)\\
\bottomrule
\end{tabular}
\end{center}

\section{Predictions fixed before the runs}
\label{app:prereg}
Decision rules and the following directional predictions were written down before
each experiment was run.

\begin{center}\small
\begin{tabular}{p{0.52\textwidth}lp{0.28\textwidth}}
\toprule
Prediction & Where & Outcome\\
\midrule
Direct-response $r(\hat\tau,\tau)$ rises with $\sigma_\eta$ & Sec.~5, Supp.~D & Held ($0.24\to0.39$)\\
Direct-response NMAE $\le$ pre-period at every $\sigma_\eta$ & Supp.~D & Held for $\sigma_\eta\ge0.6$; pre-period lower by $0.7$, $0.2$ at $0.1$, $0.3$, within about one seed s.d.\\
Direct-response $r$ $>$ pre-period $r$ at every $\sigma_\eta$ & Supp.~D & Held\\
GATv2$+$self within two points of GraphSAGE (synthetic) & Supp.~F & Exceeded: $0.4$--$0.9$ below SpilloverNet\\
If $\rho_\ell<1$ matters at high degree, GATv2$+$self trails SpilloverNet on Flickr & Supp.~F & Held ($3.1$ points)\\
CATE targeting regret $\ge$ degree-only at every $\sigma_\eta$ & Table~2, Supp.~G & Held for $\sigma_\eta\ge0.3$; reversed at $0.1$ (tie at $30\%$)\\
SpilloverNet $\approx$ Mean on the constant-$\tau$ generator of \citet{zhang2025} & Supp.~C & Held ($1.33$ vs $1.85$, $p=0.41$)\\
\bottomrule
\end{tabular}
\end{center}

\end{document}